\documentclass[twocolumn,12pt,conference]{IEEEtran}
\usepackage{epsfig,endnotes}
\pdfoutput=1

\usepackage{color}
\usepackage[table]{xcolor}
\usepackage[cmex10]{amsmath} 
\usepackage{amssymb}
\usepackage{bbm} 
\usepackage{graphicx} 
\usepackage[caption=false,font=footnotesize]{subfig} 
\usepackage[rightcaption]{sidecap} 
\usepackage{tabularx}
\newcolumntype{Y}{>{\centering\arraybackslash}X}
\usepackage{fixltx2e} 
\usepackage{engord} 
\usepackage{datetime} 
\usepackage{enumitem}
\usepackage{xifthen} 
\usepackage{multicol} 
\usepackage{amsthm} 
\usepackage[normalem]{ulem}
\usepackage[group-separator={,}]{siunitx}
\usepackage{placeins}
\usepackage{thm-restate}
\usepackage{framed}

\usepackage[linesnumbered,ruled]{algorithm2e}
\SetKwBlock{KwFunction}{function}{}
\SetKwBlock{KwOn}{on}{}
\SetKwBlock{KwPeriodically}{periodically}{}
\SetKwBlock{KwState}{state}{}
\SetKwBlock{KwAtomic}{atomically}{}
\SetKwBlock{KwConcurrently}{concurrently}{}
\SetCommentSty{textup} 
\SetKwComment{KwIttayComment}{(}{)} 
\SetKwComment{KwTitle}{\bf}{} 
\SetKw{KwAnd}{and}
\SetKw{KwOr}{or}
\SetKwFor{ForAll}{forall}{}

\newcommand\footnoteWithoutMarker[1]{%
  \begingroup
  \renewcommand\thefootnote{}\footnote{#1}%
  \addtocounter{footnote}{-1}%
  \endgroup
}

\newtheorem{definition}{Definition}
\newtheorem{lemma}{Lemma}

\newtheorem{claim}{Claim}

\newcommand{\negspace}{\vspace{-0.5\baselineskip}}

\title{ 
Cache Serializability: \\ 
Reducing Inconsistency in Edge Transactions
} 

\negspace 
\negspace 

\author{\rm Ittay Eyal \qquad Ken Birman \qquad Robbert van Renesse \\ \\ Cornell University} 

\begin{document} 

\maketitle

\pagestyle{empty}


\footnoteWithoutMarker{ 
Ittay Eyal; Ken Birman; Robbert van Renesse, ``Cache Serializability: Reducing Inconsistency in Edge Transactions,'' Distributed Computing Systems (ICDCS), IEEE 35th International Conference on, June~29 2015--July~2 2015. 

This work was supported, in part, by research grants from DARPA under its MRC program, and NSF under its Cloud Computing program. 
} 

\newcommand{\tCache}{\mbox{T-Cache}} 

%
%

\begin{abstract} 
Read-only caches are widely used in cloud infrastructures to reduce access latency and load on backend databases. 
Operators view coherent caches as impractical at genuinely large scale and many client-facing caches are updated in an asynchronous manner with best-effort pipelines. 
Existing solutions that support cache consistency are inapplicable to this scenario since they require a round trip to the database on every cache transaction. 

Existing incoherent cache technologies are oblivious to transactional data access, even if the backend database supports transactions. We propose \emph{\tCache}, a novel caching
policy for read-only transactions in which inconsistency is tolerable (won't cause safety violations) but undesirable (has a cost). \tCache\ improves cache consistency despite asynchronous and unreliable communication between the cache and the database. 
We define \emph{cache-serializability}, a variant of serializability that is 
suitable for incoherent caches, and prove that with unbounded resources 
\tCache\ implements this new specification.  
With limited resources, \tCache\ allows the system manager to choose a trade-off between performance and consistency. 

Our evaluation shows that \tCache\ detects many inconsistencies with only nominal overhead. 
We use synthetic workloads to demonstrate the efficacy of \tCache\ when data accesses are clustered and its adaptive reaction to workload changes. 
With workloads based on the real-world topologies, \tCache\ 
detects $43-70\%$ of the inconsistencies and increases the rate of consistent transactions by $33-58\%$. 
\end{abstract}

\negspace
    \section{Introduction} 

Internet services like online retailers and social networks store important data sets in large distributed
databases. 
Until recently, technical challenges have forced such large-system
operators to forgo transactional consistency, providing per-object
consistency instead, often with some form of eventual consistency.
In contrast, backend systems often support transactions 
with guarantees such as snapshot isolation and
even full transactional atomicity~\cite{corbett2013spanner,
balakrishnan2013tango, eyal2013ordering, escriva2013warp}. 

Our work begins with the observation that
it can be difficult for client-tier applications to
leverage the transactions that the databases provide: transactional reads
satisfied primarily from edge caches cannot guarantee coherency.  Yet, by running from cache, client-tier transactions shield the backend database
from excessive load, and because caches are typically placed close
to the clients, response latency can be improved.  The result is a tradeoff: we run transactions against incoherent caches, and each time a cache returns inconsistent data,
the end-user is potentially exposed to visibly incorrect results.

The problem centers on the asynchronous
style of communication used between the database and the geo-distributed
caches.  
A cache must minimize the frequency of backend database interactions; any approach requiring a significant rate of round-trips to the database would
incur unacceptable latency.
A cache must also respond promptly, which rules out asynchronous update policies that might
lock cache entries.
The latency from cache to database and back is often high, ruling out 
cache coherency schemes that would require the backend
database to invalidate or update cached objects.  In many settings the backend can't even track the locations at which cached objects might reside. 
Here, we define a variant of serializability called cache-serializability that is suitable for incoherent caches. 

Today, many transactional 
web applications, from social networks to online retailers, run over caches oblivious to consistency. 
These transactions won't crash if an inconsistency arises, but the end-user may be frustrated and platform revenue is subsequently reduced. 
With \tCache, the benefits of reading from an edge cache are maintained, but the frequency of these negative costs
is reduced. Our protocol is especially beneficial for workloads where data accesses are clustered, which is common in today's large-scale systems. 
\tCache\ achieves this benefit by storing dependency information with the
cached objects, allowing the cache (or the application)
to identify possible inconsistencies without contacting the database.
The user can improve the level of consistency by adjusting
the size of this dependency data: more dependency data leads to
increased consistency. 

To demonstrate the efficacy of the proposed scheme, we
created a prototype implementation and exposed it to workloads
based on graphically-structured real-world data, such as those seen
in social-networking situations.  The method detects $43-70\%$ of 
the inconsistencies and can increase the ratio of consistent transactions by $33-58\%$, 
both with low overhead.  We construct synthetic workloads and explore the behavior of
\tCache\ with different degrees of data clustering, and also investigate its rate of response when clusters change. 

With perfectly clustered workloads, \tCache\ implements full cache-serializability. 
To explain this perfect behavior we prove a related claim~--- we show that with unbounded resources 
\tCache\ implements cache-serializability. 

In summary, the contributions of this work are: 
\begin{enumerate} 

\item Definition of cache-serializability, a variant of serializability suitable for incoherent caches. 

\item The \tCache\ architecture, which allows trading off efficiency and transaction-consistency in large scale cache deployments. 

\item Evaluation of \tCache\ with synthetic workloads, demonstrating its adaptivity and sensitivity to clustering. 

\item Evaluation of \tCache\ with workloads based on graphically-structured real-world data. 

\item Proof that \tCache\ with unbounded resources implements cache-serializability. 

\end{enumerate}

    \section{Motivation} 





\paragraph{Two-tier structure} 
Large Internet services store vast amounts of data. Online retailers
such as Amazon and eBay maintain product stocks and information,
and social networking sites such as Facebook and Twitter maintain
graphical databases representing user relations and group structures.
For throughput, durability, and availability, such databases are
sharded and replicated.

The vast majority of accesses are read-only (e.g., Facebook reports
a~99.8\% read rate~\cite{bronson2013tao}).  To  reduce
database load and to reduce access latency, these companies employ
a two-tier structure, placing layers of cache servers in front of the database
(see Figure~\ref{fig:dbAndCache}).

The caches of primary interest to us are typically situated far from the backend
database systems~--- to reduce
latency, companies place caches close
to clients.  Timeouts are used to ensure that stale cached objects
will eventually be flushed, but to achieve a high cache hit ratio,
timeout values are generally large.  To obtain reasonable consistency,
the database sends an asynchronous stream of invalidation records or
cache updates, often using protocols optimized for throughput and freshness
and lacking absolute guarantees of order or reliability.

It is difficult to make this invalidation mechanism reliable without
hampering database efficiency.  First, the
databases themselves are large, residing on many servers, and may be geo-replicated.
Databases supporting cloud applications are often pipelined, using locks prudently in order to maximize concurrency
but ensuring that there is always other work to do; the result is that some updates complete rapidly but others
can exhibit surprisingly high latency.  During these delays, read access must go forward.  
Databases cannot always accurately track the caches that hold a
copy of each object, because the context makes it very hard to send timely invalidations: 
they could be delayed (e.g., due to buffering or retransmissions
after message loss), not sent (e.g., due to an inaccurate list of
locations), or even lost (e.g., due to a system configuration change, 
buffer saturation, or because
of races between reads, updates, and invalidations).
A missing invalidation obviously leaves the corresponding cache entry stale.
Pitfalls of such invalidation schemes are described in detail by Nishita et
al.~\cite{nishtala2013scaling} and by Bronson et al.~\cite{bronson2013tao}.

\begin{figure}[!t] 
\centering
\includegraphics[width=0.8\linewidth]{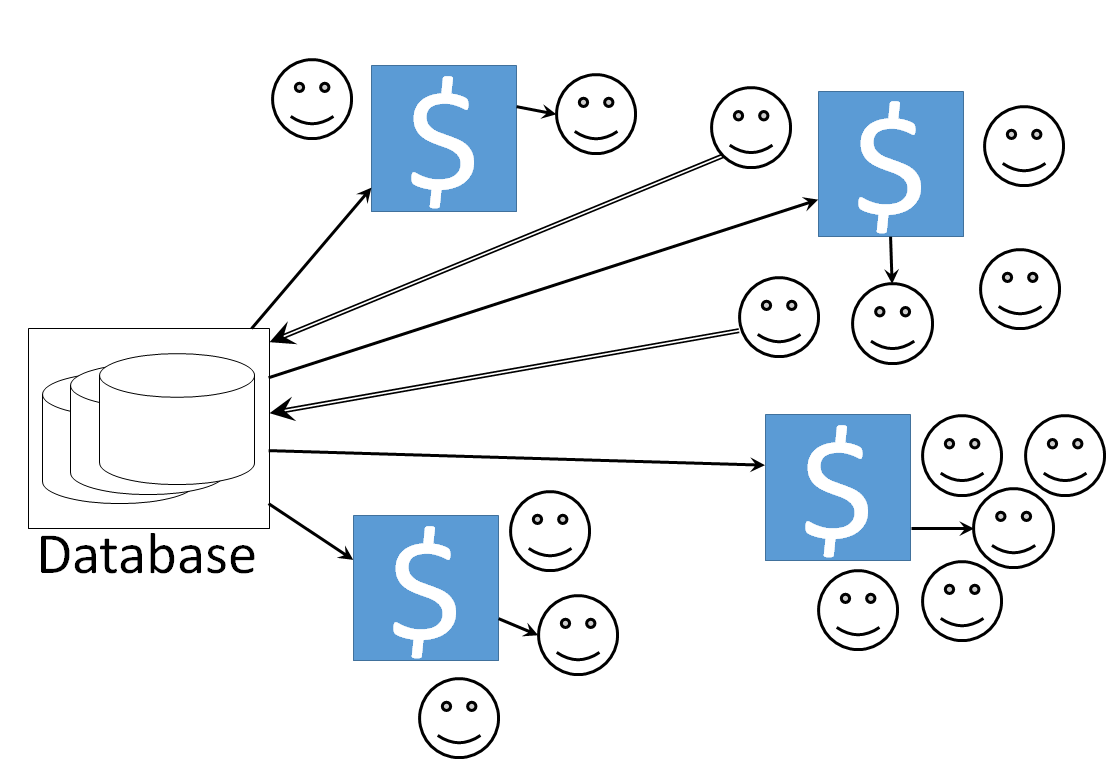} 
\caption{ 
The common two-tiered structure. 
Clients perform read-only transactions by accessing caches, which receive their values by reading from the database (solid lines). Update transactions go directly to the database (double lines). Subsequent cache invalidations can be delayed or even lost due to race conditions, leading to a potentially inconsistent view by the cache clients. 
} 
\label{fig:dbAndCache} 
\end{figure}

\paragraph{DB Transactional consistency} 
The complexity of implementing geo-scale databases with strong
guarantees initially led companies to abandon cross-object consistency
altogether and make do with weak guarantees such as per-object
atomicity or eventual consistency.  In effect, such systems do
repair any problems that arise, eventually, but the end-user is
sometimes exposed to inconsistency.
For some applications this is acceptable, and the approach
has been surprisingly successful.  In today's cloud, relaxed consistency is
something of a credo.

But forgoing transactional consistency can result in undesired
behavior of a service.  Consider a buyer at an online site who looks
for a toy train with its matching tracks just as the vendor is
adding them to the database.  The client may see only the train in
stock but not the tracks because the product insertion transaction would often be
broken into two or more atomic but independent sub-transactions.
In a social network, an inconsistency with unexpected results can
occur if a user $x$'s record says it belongs to a certain group,
but that group's record does not include $x$.
Web albums maintain picture data and access control lists (ACLs)
and it is important that ACL and album updates are consistent (the
classical example involves removing one's boss from the album ACL
and then adding unflattering pictures).

As a result, we see more and more systems that are basically transactional in design and structure,
but incorporate workarounds to tolerate weak consistency.  When inconsistency does arise, the platform's 
utility is reduced. This is our target: our goal is to reduce the frequency of such events as much as possible.
There has been a wave of recent innovations within the backend,
offering scalable object stores that can efficiently support transactions through snapshot isolation and even full atomicity~\cite{corbett2013spanner, balakrishnan2013tango, eyal2013ordering, escriva2013warp}. 
Our approach seeks to improve transaction consistency at the cache 
layer, without requiring the cache to access the backend on each read or to invalidate items that may be undergoing updates. 


\paragraph{Caches bring value but cause problems}
As noted, many of today's systems are transactional at the backend, yet
inconsistent because they use edge caching decoupled from the internal transactional mechanisms.  Even when the database itself is consistent, the vast majority
of operations are read-only transactions issued by edge clients and
are at high risk of observing inconsistent state in the cache.  The
outright loss of cache invalidations emerges as an especially
significant problem if transactional consistency is of value.
In our work, we start by accepting that any solution must preserve
performance properties of the existing caching tier.  In particular, we
need to maintain the shielding role of the cache: the cache hit
ratio should be high.  Additionally, a read-only cache access should
complete with a single client-to-cache round-trip on cache hits.
This rules out coherent cache solutions such as~\cite{ports2010txCache}.
On the other hand, our approach gains leverage by adopting the view that inconsistency
is undesirable but not unacceptable.  Our goal is thus to {\em reduce the rate of user-visible inconsistencies.}


    \section{Architecture} 

Since the cache is required to respond immediately to the client on hits, and the database-cache channel is  asynchronous, we decided to employ
a transactional consistency that is weaker than the full ACID model. In our approach, read-only transactions and update transactions that access the same cache
 are guaranteed an atomic execution, but read-only transactions that access different caches may observe different orderings for  independent update transactions. 

\begin{definition}[Cache serializability] \label{cacheSerializability} 
For every execution $\sigma$, every partial execution that includes all update transactions in $\sigma$ and all read-only transactions that go through a single cache server, is serializable. 
\end{definition} 

Our solution seeks to come as close as possible to cache serializability, subject to constraints imposed by the bounded
size of the edge caches and the use of asynchronous communication with the DB. 
We start with an observation: in many scenarios, objects
form \emph{clusters} with strong locality properties. Transactions are 
likely to access objects that are, in some sense, close to each other. 
For retailers
this might involve related products, for social networks the set
of friends, for geographical services physical proximity, and for
web albums the ACL objects and the pictures assigned to them.  
Moreover, in some cases applications explicitly cluster their 
data accesses to benefit from improved parallelism~\cite{XWB13}. 
The resulting
transactions access objects from a single cluster, although there will
also be some frequency of 
transactions that access unrelated objects in different clusters.

Our solution requires minor changes to the database object
representation format, imposing a small and constant memory overhead
(that is, independent of the database
size and the transaction rate).  This overhead involves tracking and caching what we refer
to as {\em dependency lists.} These are bounded-length lists of object identifiers and
the associated version numbers, each representing some recently
updated objects upon which the cached object
depends.  

A bounded-sized list can omit dependency information required to detect inconsistencies, hence it is
important to use a bound large enough to capture most of the relevant dependencies.
At present we lack an automated way to do this: we require the
developer to tune the length so that the frequency of errors is
reduced to an acceptable level, reasoning about the trade-off (size versus accuracy) in a manner we discuss further below.  
Intuitively, dependency lists should be roughly the same size as the size of the 
workload's clusters. 

Our extensions offer a transactional interface to the cache
in addition to the standard read/write API.
In many cases, our algorithm detects and fixes inconsistent read-only
transactions at the cache with constant complexity.  It does so by
either aborting the transaction (which can then be retried), or
invalidating a cached object which can then force a read from the
database (similar to handling cache misses).  When the dependency lists 
fail to document a necessary dependency, an
application might be exposed to stale values. 

Because we have in mind client-side applications that are unlikely to 
validate against the back-end, for many of our intended uses some level of undetected inconsistency can slip past.
However, because the developer would often be able to tune the mechanism, during
steady-state operation of large applications, the rate of unnoticed inconsistencies could be extremely low. 

With clustered workloads we will demonstrate that it is sufficient to
store a small set of dependencies to detect most inconsistencies.  We
also investigate workloads where the clustered access pattern is
less strongly evident; here, our approach is less effective even with longer
dependency list lengths.  Thus our solution is not a panacea, but, for applications
matched to our assumptions, can be highly effective.

        \subsection{Database} 

We assume that the database tags each object with a version number specific to the transaction that
most recently updated it, and that there is a total ordering on version numbers. The version of a
transaction is chosen to be larger than the versions of all objects accessed by the transaction.
The database stores for each object $o$ a list of $k$
dependencies $(d^o_1, v^o_1), (d^o_2, v^o_2), \dots (d^o_k, v^o_k)$.
This is a list of identifiers and versions of other objects that the current version
of $o$ depends on. A read-only transaction that sees the current version of $o$ must not
see object $d_i$ with version smaller than $v_i$. 

When a transaction $t$ with version $v_t$ touches objects $o_1$ and
$o_2$, it updates both their versions and their dependency lists.
Subsequent accesses to object $o_1$ must see object $o_2$ with a
version not smaller than $v_t$. Moreover, it inherits all of the~$l$
dependencies of $o_2$ (where $l$ is the length of $o_2$'s dependency
list).  So the dependency list of $o_1$ becomes
\begin{multline*} 
(d^{o_1}_1, v^{o_1}_1), (d^{o_1}_2, v^{o_1}_2), \dots (d^{o_1}_k, v^{o_1}_k), \\
(o_2, v_t), 
(d^{o_2}_2, v^{o_2}_2), (d^{o_2}_3, v^{o_2}_3), \dots (d^{o_2}_l, v^{o_2}_l) \,\,\, . 
\end{multline*} 

\newcommand{\key}{\ensuremath{ \textit{key} }}
\newcommand{\ver}{\ensuremath{ \textit{ver} }}
\newcommand{\depList}{\ensuremath{ \textit{depList} }}
\newcommand{\readSet}{\ensuremath{ \textit{readSet} }}
\newcommand{\writeSet}{\ensuremath{ \textit{writeSet} }}

When a transaction is committed, this update is done for all objects in the transaction at once. Given a read set \readSet, and a write set \writeSet, containing tuples comprised of the keys accessed, their versions and their dependency lists, the database aggregates them to a single full dependency list as follows
\[
\textit{full-dep-list} 
\gets 
\hspace{-0.1in}
\bigcup_{\substack{(\key, \ver, \depList) \in \\ \readSet \cup \writeSet}} 
\hspace{-0.1in}
    \{(key,ver)\} 
    \cup 
    \depList \,\,\, .
\]

This list is pruned to match the target size using LRU, and stored with each write-set object. 
A list entry can be discarded if the same entry's object appears
in another entry with a larger version.  Nevertheless, were their lengths not bounded,
dependency lists could quickly grow to include all objects in the database.

        \subsection{Cache}\label{sec:cache}

In our scheme, the cache interacts with the database in essentially
the same manner as for a consistency-unaware cache,  performing
single-entry reads (no locks, no transactions) and receiving invalidations
as the database updates objects. Unlike consistency-unaware caches,
the caches read from the database not only the object's value, but
also its version and the dependency list. 

To its clients, the extended  cache exports a transactional read-only
interface. Client read requests are extended with a
transaction identifier and a last-op flag 
\[
\texttt{read(txnID, key, lastOp)} \,\,\, . \hspace{\linewidth}
\]
The transaction identifier \texttt{txnID} allows the cache to recognize 
reads belonging to the same transaction. 
The cache responds with either the value of the requested object, or 
with an abort  if it detects an inconsistency between this read and 
any of the previous reads with the same transaction ID. 
We do not guarantee that inconsistencies will be detected. The 
\texttt{lastOp} allows the cache to garbage-collect its transaction 
record after responding to the last read operation of the transaction. 
The cache will treat subsequent accesses with the same transaction ID 
as new transactions. 

To implement this interface, the cache maintains a record of each
transaction with its read values, their versions, and their dependency
lists. On a read of~$\key_\textit{curr}$, the cache first obtains the requested entry from memory (cache hit), or database (cache miss). 
The entry includes the value, version~$\ver_\textit{curr}$ and dependency list $\depList_\textit{curr}$. 
The cache checks the currently read object against each of the previously read objects. 
If a previously read version $v'$ is older than the version $v$ expected by the current read's dependencies 
\begin{multline} \label{eqn:prevOld}
\exists k, v, v': v > v' \wedge (k, v) \in \depList_\textit{curr} \wedge \\
(k, v') \in readSet \cup \writeSet \,\,\, , 
\end{multline} 
or the version of the current read~$v_\textit{curr}$ is older than the version~$v$ expected by the dependencies of a previous read 
\begin{equation} \label{eqn:curOld}
\exists v: v > v_\textit{curr} \wedge (\key_\textit{curr}, v) \in readSet \cup \writeSet \,\,\, , 
\end{equation}
an inconsistency is detected. Otherwise the cache returns the
read value to the client. 

Upon detecting an inconsistency, the cache can take one of three
paths:

\begin{enumerate}
\item \texttt{ABORT}: abort the current transaction.  Compared to
the other approaches, this has the benefit of affecting only the
running transaction and limiting collateral damage.
\item \texttt{EVICT}:
abort the current transaction \emph{and} evict the violating (too-old) object
from the cache. This approach guesses that 
future transactions are likely to abort because of this object. 
\item \texttt{RETRY}: check which is the violating object. If it is the currently accessed object (Equation~\ref{eqn:curOld}), treat this access as a miss and
respond to it with a value read from the database. If the violating
object was returned to the user as the result of a read earlier in the transaction (Equation~\ref{eqn:prevOld}),
evict the stale object and abort the transaction (as in \texttt{EVICT}). 
\end{enumerate}

        \subsection{Consistency} 

With unbounded resources, \tCache\ detects all inconsistencies, as stated in the following theorem. 

\begin{restatable}{thm}{consistencyThm} \label{thm:consistency} 
\tCache\ with unbounded cache size and unbounded dependency lists implements cache-serializability. 
\end{restatable} 

The proof (deferred to Appendix~\ref{app:theory}) is by constructing a serialization of the transactions in the database and in one cache, based on the fact that the transactions in the database are serializable by definition. 

The implications of Theorem~\ref{thm:consistency} will be seen in Section~\ref{sec:eval:temporal}. 
\tCache\ converges to perfect detection when stable clusters are as large as its dependency lists. In such a scenario, the dependency lists are large enough to describe all relevant dependencies.

    \section{Experimental Setup} 

To evaluate the effectiveness of our scheme, we implemented a
prototype.  To study the properties of the cache, we only need a
single column (shard) of the system, namely a single cache backed
by a single database server.

\begin{figure}[!t] 
\centering 
\includegraphics[width=0.8\linewidth]{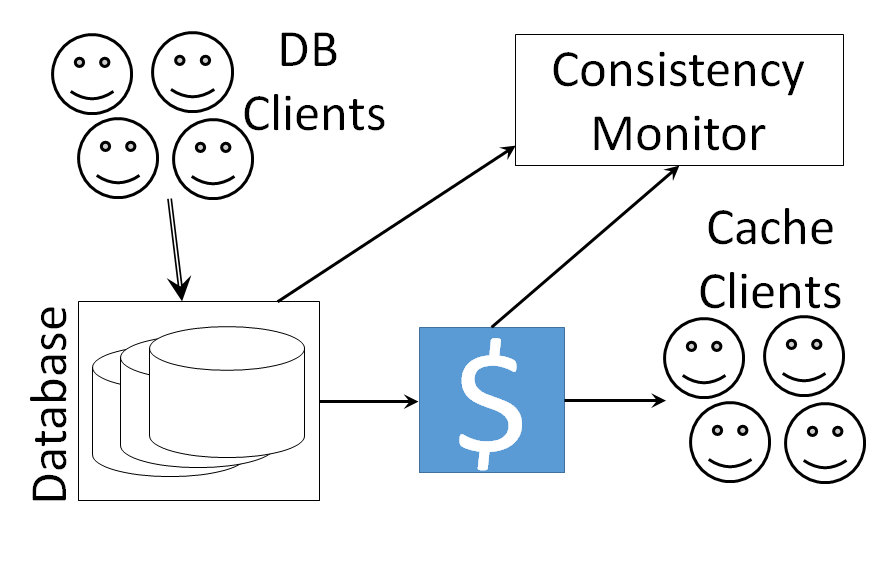} 
\caption{ 
Experimental setup. Update clients access database, which sends invalidations to the cache. 
Read-only clients access cache. Consistency monitor (experiment-only element) receives all transactions and rigorously detects inconsistencies for statistics. 
} 
\negspace
\label{fig:expScheme} 
\end{figure}

Figure~\ref{fig:expScheme} illustrates the structure of our
experimental setup.  A single database implements a transactional
key-value store with two-phase commit.
A set of cache clients perform read-only transactions through a
single cache server.  The cache serves the requests from its local
storage if possible, or reads from the database otherwise.

On startup, the cache registers an upcall that can be used by the database to report invalidations; 
after each update transaction, the database asynchronously sends
invalidations to the cache for all objects that were modified.  A ratio of
$20\%$ of the invalidations, chosen uniformly at random, are dropped
by the experiment; this is extreme and would only be seen in the real world under
conditions of overload or when the system configuration is changed.

Both the database and the cache report all completed transactions
to a consistency monitor, created in order to gather statistics for
our evaluation. This server collects both committed and aborted
transactions and it maintains the full dependency graph.  It performs
full serialization graph testing~\cite{bernstein1987concurrency}
and calculates the rate of inconsistent transactions that committed
and the rate of consistent transactions that were unnecessarily
aborted.

Our prototype does not address the issue of cache eviction when
running out of memory.  In our experiments, all objects in the
workload fit in the cache, and eviction is only done if there is
a direct reason, as explained below.  Had we modeled them, evictions would 
reduce the cache hit rate, but could not cause new inconsistencies.

We evaluate the effectiveness of our transactional cache
using various workloads and varying the size of the dependency lists
maintained by the cache and the database. 
For the cases considered, short dependency lists suffice (up to~5 versions per object). 
An open question for further study is whether there are workloads that might require limited but larger values. 
Note that dependencies arise from the topology of the object graph, and 
not from the size of the transactions' read and write sets. 
        
As a baseline for
comparison, we also implemented a timeout-based approach: it reduces
the probability of inconsistency by limiting the life span of cache
entries. We compare this method against our transactional cache by
measuring its effectiveness with a varying time-to-live (TTL) for
cache entries. 

In all runs, both read and update transactions access~5 objects per
transaction. Update clients access the database at a rate of~100
transactions per second, and read-only clients access the cache at
a rate of~500 transactions per second. 

Our experiment satisfies all read-only transactions from the cache, 
while passing all update transactions directly to the backend database.  
Each cache server is unaware of the other servers~--- it has its 
own clients and communicates directly with the backend database. 
The percentage of read-only transactions can be arbitrarily 
high or low in this situation: with more 
caches, we can push the percentage up. 
Our simulation focuses on just 
a single cache---it would behave the same had there been 
many cache servers.

    \section{Evaluation} \label{sec:evaluation} 

\tCache\ can be used with any transactional backend and any transactional
workload.  Performance for read-only transactions will be similar to
non-transactional cache access: the underlying database is only accessed
on cache misses.  However, inconsistencies may be observed.

First, we will use synthetic workloads so we can evaluate how much
inconsistency can be observed as a function of the amount of clustering
in the workload.  This also allows us to look at the dynamic behavior of the system,
when the amount of clustering and the clustering formation change over time. 

Next, we will look at workloads based on Amazon's product co-purchasing and
Orkut's social network to see how much inconsistency \tCache\ can detect
as a function of dependency list length, and compare this with a TTL-based
approach.  We are also interested in overhead, particularly the additional
load on the backend database that could form if the the rate of cache misses 
increases.

Section~\ref{sec:cache}
presented three strategies for responding to inconsistency detection.
For both the synthetic and realistic workloads, we compare
the efficacy of the three strategies.

        \subsection{Synthetic Workloads} \label{sec:evaluation-artificial} 

Synthetic workloads allow us to understand the efficacy of \tCache\ as
a function of clustering.
For the experiments described here, we configured \tCache\ with a maximum
of~5 elements per dependency list. 


Section~\ref{sec:artificialWorkload} describes synthetic workload generation.
Section~\ref{sec:eval:clustering} measures how many inconsistencies
we can detect as a function of clustering and Section~\ref{sec:eval:temporal}
considers clustering changes over time.
Section~\ref{sec:art-detect-prevent} compares the efficacy of various
approaches to dealing with detected inconsistencies.

            \subsubsection{Synthetic Workload Generation} \label{sec:artificialWorkload} 

Our basic synthetic workload is constructed as follows. 
We use~2000 objects numbered~0 through~1999. The objects are divided into clusters of size~5: $0-4, 5-9, 10-14, \dots$, and there are two types of workloads. In the first, clustering is perfect and each transaction chooses a single cluster and chooses~5 times with repetitions within this cluster to establish its access set. In the second type of workloads access is not fully contained within each cluster. When a transaction starts, it chooses a cluster uniformly at random, and then picks~5 objects as follows. Each object is chosen using a bounded Pareto distribution starting at the head of its cluster $i$ (a product of~5). If the Pareto variable plus the offset results in a number outside the range (i.e., larger than 1999), the count wraps back to~0 through~$i-1$. 

            \subsubsection{Inconsistency Detection as a Function of $\alpha$} \label{sec:eval:clustering} 

We start by exploring the importance of the cluster structure by
varying the $\alpha$ parameter of the Pareto distribution. We
vary the Pareto $\alpha$ parameter from~1/32 to~4.  In this experiment
we are only interested in detection, so we choose the ABORT strategy. 

Figure~\ref{fig:paretoAlpha} shows the ratio of inconsistencies
detected by \tCache\ compared to the total number of potential
inconsistencies. 
At $\alpha=1/32$, the distribution is almost uniform
across the object set, and the inconsistency detection ratio is
low~--- the dependency lists are too small to hold all relevant
information. At the other extreme, when~$\alpha=4$, the distribution
is so spiked that almost all accesses of a transaction are within
a cluster, allowing for perfect inconsistency detection.  We
note that the rate of detected inconsistencies is so high at this
point that much of the load goes to the backend database and saturates
it, reducing the overall throughput.

\begin{figure}[!t]
\centering
\includegraphics[width=0.8\linewidth]{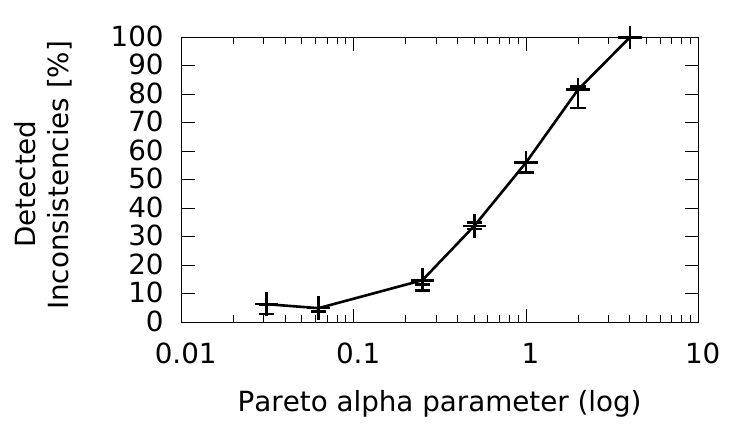}
\caption{Ratio of inconsistencies as a function of $\alpha$.}
\label{fig:paretoAlpha}
\end{figure}

            \subsubsection{Convergence} \label{sec:eval:temporal} 

So far we have considered behavior with static clusters, that
is, over the entire run of each experiment accesses are confined
to the same (approximate) clusters.  Arguably, in a real system,
clusters change slowly, and so if \tCache\ converges to maintain the
correct dependency lists as clusters change, our setup 
serves as a valid quasi-static analysis. 

In this section, we investigate the convergence of \tCache\ when
clusters change over time.  Since the dependency lists of the objects are
updated using LRU, the dependency list of an object $o$ tends
to include those objects that are frequently accessed together with~$o$.  
Dependencies in a new cluster automatically push out dependencies
that are now outside the cluster.

            \subsubsection*{Cluster formation} 

To observe convergence, we perform an experiment where accesses
suddenly become clustered.  Initially accesses are uniformly at
random from the entire set (i.e., no clustering whatsoever), then
at a single moment they become perfectly clustered into 
clusters of size~5. Transactions 
are aborted on detecting an inconsistency.  We use a transaction rate 
of approximately~\num{500}
per second. The database includes \num{1000} objects. 

\begin{figure}[!t] 
\centering
\includegraphics[width=0.8\linewidth]{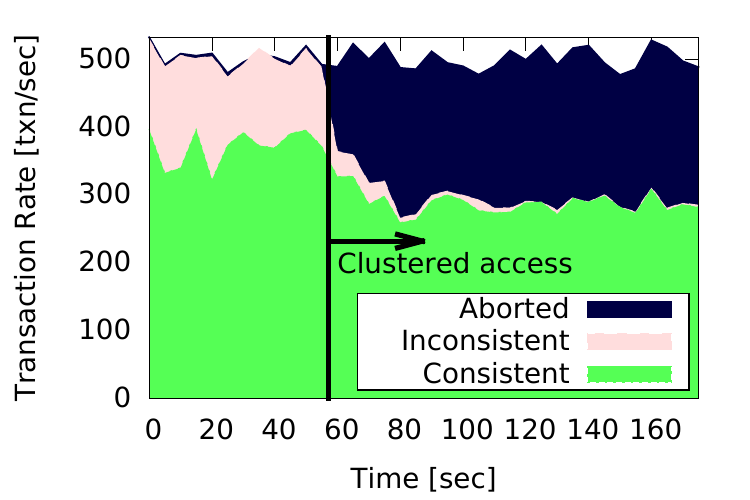} 
\caption{
Convergence of \tCache.  Before time $t = 58s$ accesses are uniformly
at random.  Afterward, accesses are clustered.
} 
\label{fig:convergence} 
\end{figure}

Figure~\ref{fig:convergence} shows the percentage of transactions
that commit and are consistent (at the bottom), the percentage of
transactions that commit but are inconsistent (in the middle), and
the percentage of transactions that abort (at the top).  Before
$t=58s$ access is unclustered, and as a result the dependency lists
are useless; only few inconsistencies are detected, that is, about~26\%
of the transactions that commit have witnessed inconsistent data.
At $t=58s$, accesses become perfectly clustered.  As desired, we
see fast improvement of inconsistency detection.
The inconsistency rate drops as the abort rate rises~--- this is desired as well.
The overall rate of consistent committed transactions drops because
the probability of conflicts in the clustered scenario is higher. 

            \subsubsection*{Drifting Clusters} 

To illustrate more realistic behavior, we use clustered accesses that slowly drift. Transactions are perfectly clustered, as in the previous experiment, but every~3 minutes the cluster structure shifts by~1 ($0-4, 5-9, 10-14 \rightarrow 1-4, 5-10, 11-15$, and wrapping back to zero after 1999).
Figure~\ref{fig:drift} shows the results.
After each shift, the objects' dependency lists are outdated. This leads to a sudden increased inconsistency rate that converges back to zero, until this convergence is interrupted by the next shift. 

\begin{figure}[!t] 
\centering 
\includegraphics[width=0.8\linewidth]{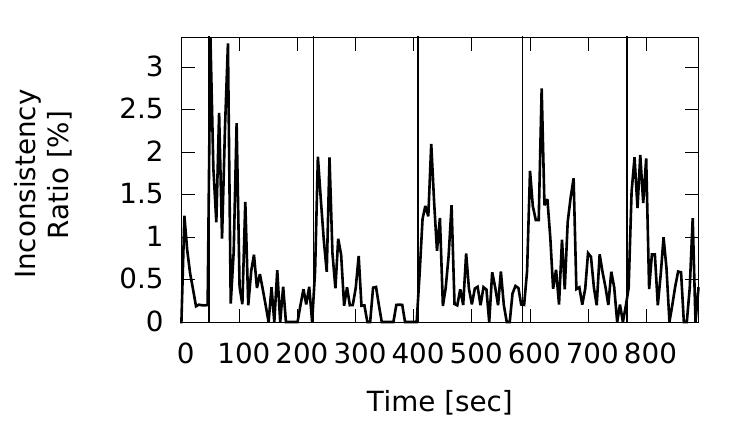} 
\caption{ 
Perfectly clustered synthetic workload where the clusters shift by~1 every~3 minutes, marked by vertical lines.
} 
\label{fig:drift} 
\end{figure} 

            \subsubsection{Detection vs. Prevention}\label{sec:art-detect-prevent}

Section~\ref{sec:cache} presented three possible strategies for the cache to deal with inconsistency detection: (1) aborting the transaction (ABORT), (2) aborting and evicting value (EVICT), and (3) read-through when possible as in cache miss, abort otherwise (RETRY).  We will now compare their efficacies.

We use the approximate clusters workload with~\num{2000} objects, a window size of~5, a Pareto~$\alpha$ parameter of~1.0, and the maximum dependency list size is set to~5. 

Figure~\ref{fig:onInconsistency-drifting} illustrates the results. For each strategy, the lower portion of the graph is the ratio of committed transactions that are consistent, the middle portion is committed transactions that are inconsistent, and the top portion is aborted transactions. 

The abort strategy provides a significant improvement over a normal, consistency-unaware cache, as the strategy detects and aborts 
over~$55\%$ of all inconsistent transactions that would have been committed. 
But the other strategies make further improvements. 
EVICT reduces uncommittable transactions to $28\%$ of its value with ABORT.
This indicates that violating (too-old) cache entries are likely to be repeat offenders: they are too old for objects that are likely to be accessed together with them in future transactions, and so it is better to evict them. 
RETRY reduces uncommittable transactions further to about $23\%$ of its value with ABORT. 

\begin{figure}[!t] 
\centering 
\includegraphics[width=0.8\linewidth]{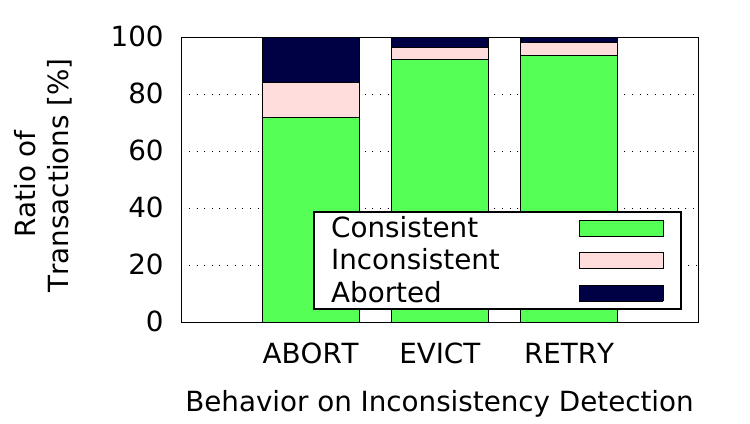} 
\caption{ 
The efficacy of \tCache\ as a function of the strategy taken for
handling detected inconsistencies. ABORT detects and aborts transactions that glimpsed inconsistent data.
EVICT aborts, but also evicts the inconsistent cached data, and RETRY performs a read-through if it resolves an inconsistent read. 
In each case we break down the transactions as consistent, undetected inconsistencies, and aborted. 
} 
\label{fig:onInconsistency-drifting} 
\end{figure}

        \subsection{Realistic Workloads} \label{sec:evaluation-realistic} 

We now evaluate the efficacy of \tCache\ with workloads based on two sampled topologies from the online retailer Amazon and the social network Orkut. 
Section~\ref{sec:realisticWorkloadGeneration} describes how we generated these workloads.
Section~\ref{sec:realResults} measures the efficacy of \tCache\ on these workloads as a function of maximum dependency list size, and compares this to a strategy based on TTLs.
Section~\ref{sec:onInconsistencyRealistic} compares the efficacy of the three strategies of dealing with detected inconsistencies.

            \subsubsection{Workload Generation} \label{sec:realisticWorkloadGeneration} 

We generated two workloads based on real data:
\begin{enumerate}
\item Amazon:
We started from a snapshot of Amazon's product co-purchasing graph
taken early~2003~\cite{leskovec2007dynamics}.  Each
product sold by the online retailer is a node and each pair of
products purchased in a single user session is an edge.
The original graph contains more than~\num{260000} nodes.
\item Orkut:
For the second, we used a snapshot of the friendship relations graph in the Orkut social network, taken late~2006~\cite{mislove2007measurement}. 
In this graph, each user is a node and each pair of users with a friend relationship is an edge. The original graph contains more than~\num{3000000} nodes.
\end{enumerate}

Because the sampled topologies are large and we only need to simulate
a single ``column'' of the system for our purposes~--- one database
server and one cache server~--- we down-sample both graphs to~\num{1000}
nodes.  We use a technique based on random walks that maintains
important properties of the original graph~\cite{leskovec2006sampling},
specifically clustering which is central to our experiment.
We start by choosing a node uniformly and random and start a random walk 
from that location. In every step, with probability~$15\%$, the walk 
reverts back to the first node and start again.  This is repeated until 
the target number of nodes have been visited.
Figure~\ref{fig:amazonGraph}(a) and (b) show a further down-sampling
to~500 nodes to provide some perception of the topologies. The graphs 
are visibly clustered, the Amazon topology more so than the Orkut one, 
yet well-connected.

Treating nodes of the graphs as database objects, transactions are
likely to access objects that are topologically close to one another.
For the online retailer, it is likely that objects bought together
are also viewed and updated together (e.g., viewing and buying a toy
train and matching rails).  For the social network, it is likely
that data of befriended users are viewed and updated together
(e.g., tagging a person in a picture, commenting on a post by a
friend's friend, or viewing one's neighborhood).

Therefore, we generate a transactional workload that accesses
products that are topologically close.
Again, we use random walks. Each transaction starts by picking a 
node uniformly at random and takes~5 steps of a random walk.  
The nodes visited by the random walk are the objects the transaction
accesses. Update transactions first read all objects from the
database, and then update all objects at the database. Read
transactions read the objects directly from the cache.

\begin{figure*}[p!] 
\centering

\subfloat[Product Affinity (Amazon)]{ 
\includegraphics[width=0.3\linewidth]{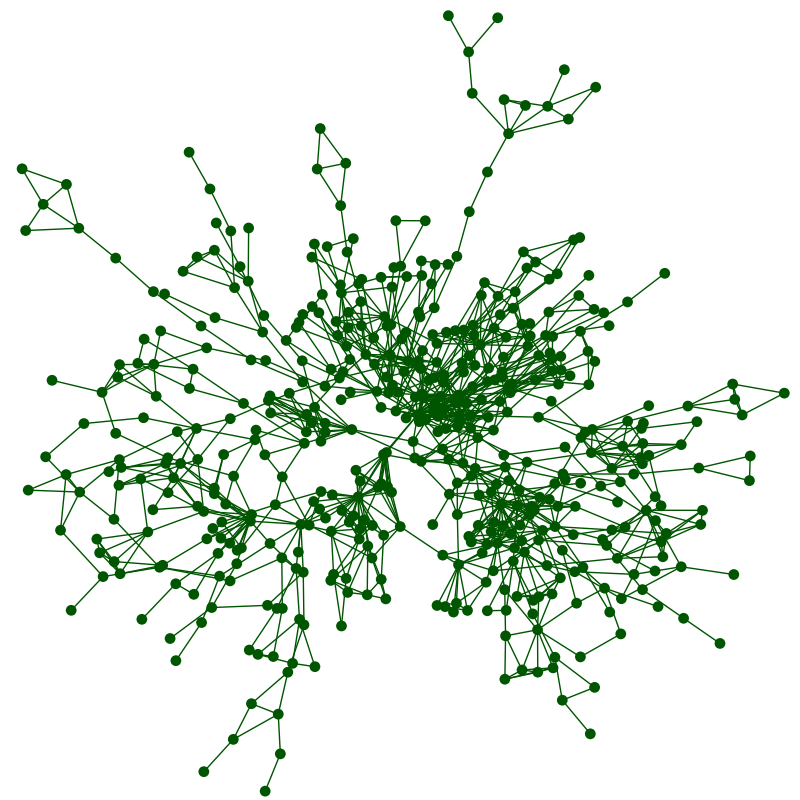} 
\label{fig:amazon:topology} 
}
\hfil
\subfloat[Social Network (Orkut)]{
\includegraphics[width=0.3\linewidth]{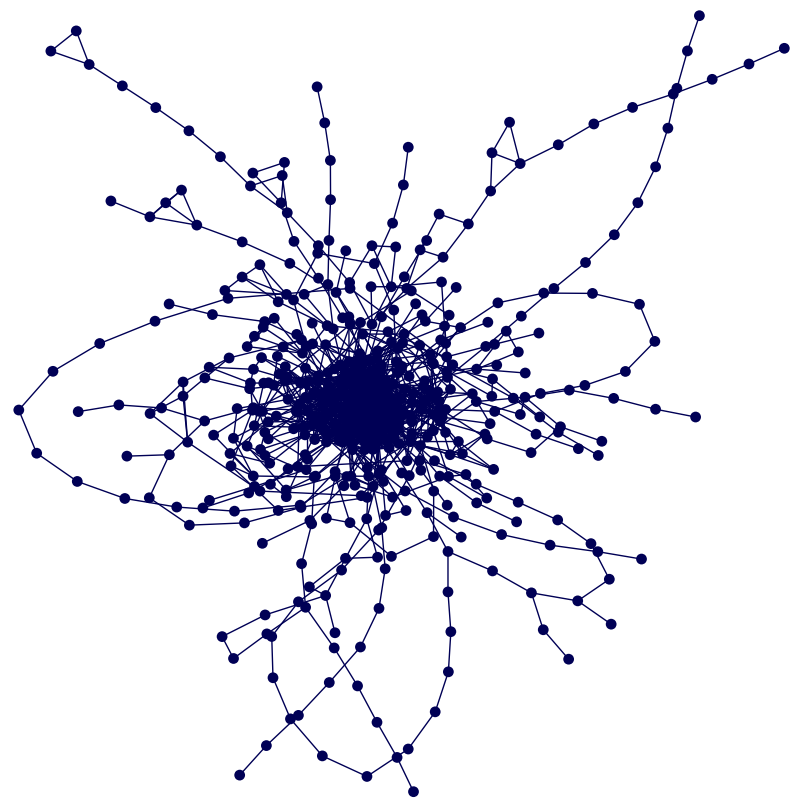} 
\label{fig:orkut:topology} 
}

\subfloat[Transactional Cache]{ 
\includegraphics[width=0.48\linewidth]{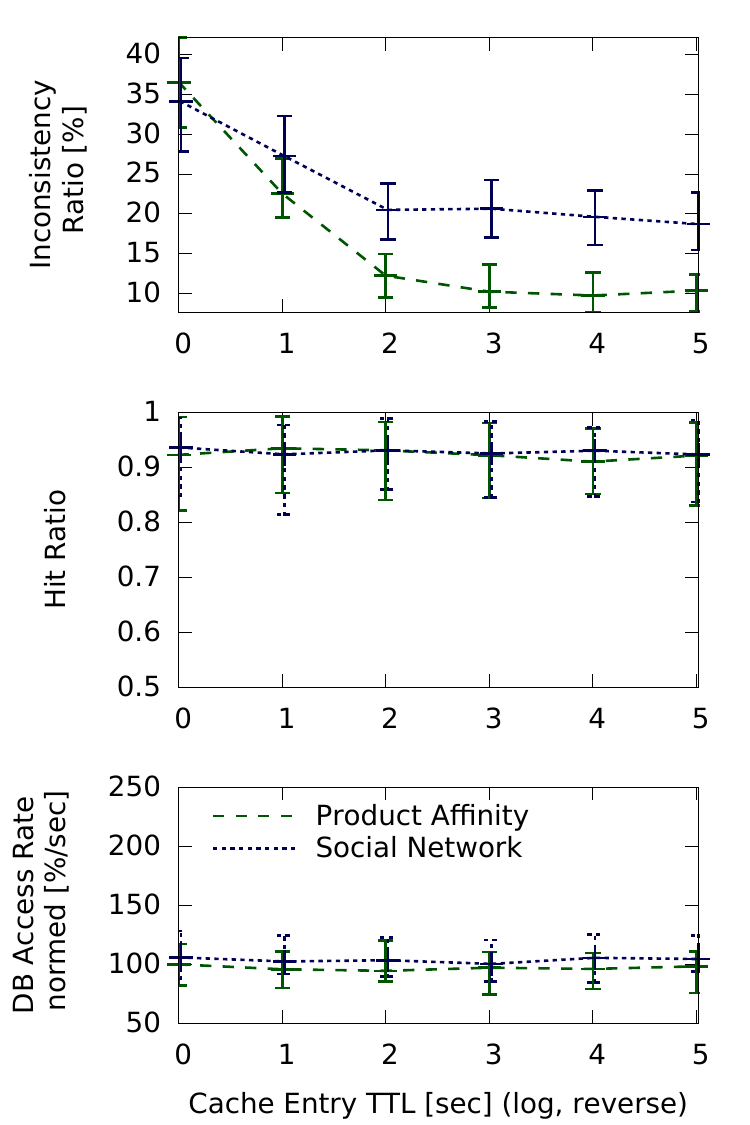} 
\label{fig:efficiency:txnCache} 
}
\hfil
\subfloat[Limited Cache Entry TTL]{ 
\includegraphics[width=0.48\linewidth]{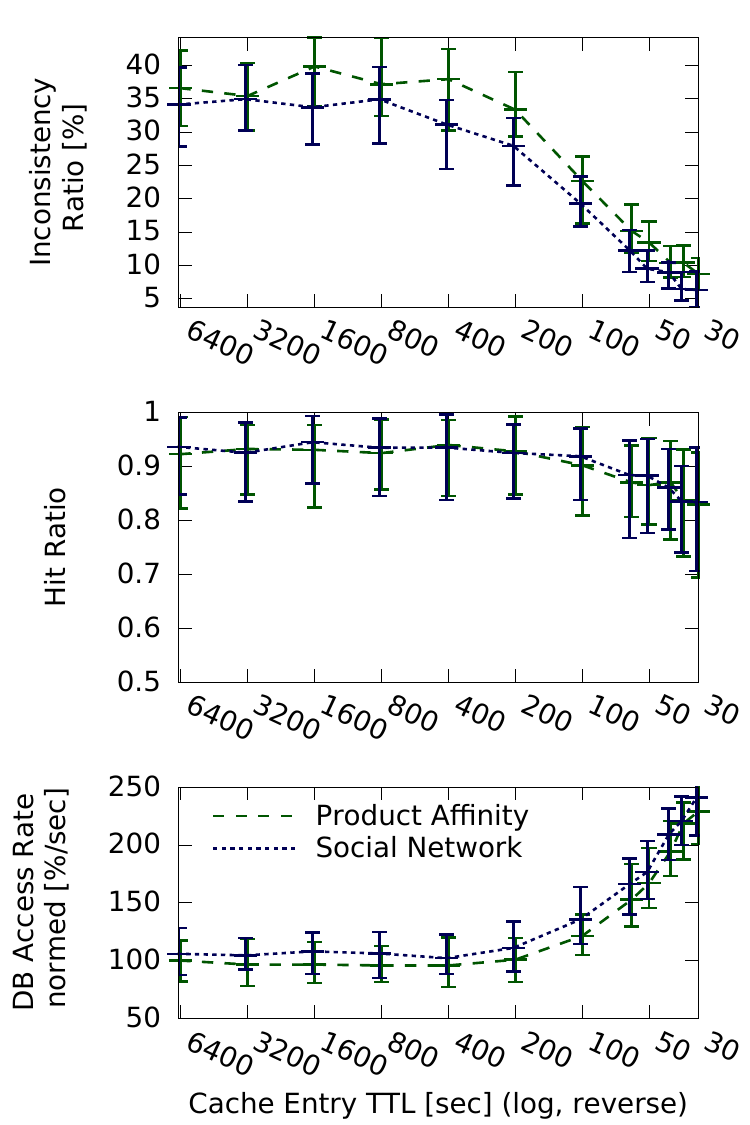} 
\label{fig:efficiency:TTL} 
} 

\caption[.]{\protect 
Experiments with workloads based on a web retailer product affinity topology and a social network topology illustrated in~\subref{fig:amazon:topology} and~\subref{fig:orkut:topology}. Transactional cache~\subref{fig:efficiency:txnCache} compared against the alternative of reducing cache entry time-to-live~\subref{fig:efficiency:TTL}. 
Data points are medians and error bars bound the 10 and 90 percentiles. 
} 
\label{fig:amazonGraph} 
\end{figure*} 

            \subsubsection{Efficacy and Overhead} \label{sec:realResults}

In this section we evaluate {\tCache} using the workloads described above.
We found that the abort rate is negligible in all runs.
Efficacy is therefore defined to be the ratio of inconsistent transactions out of all commits. 

The overhead of the system is twofold. First, dependency list maintenance implies storage and bandwidth overhead at both the database and the cache, as well~as compute overhead for dependency list merging at the server and consistency checks at the cache. However,~the storage required is only for object IDs and versions, not content, and both updates and checks are $O(1)$ in the number of objects in the system and $O(k^2)$ in the size of the dependency lists, which is limited to~5 in our experiments. 

The second and potentially more significant overhead is the effect on cache hit ratio due to evictions and hence the database load. Since cache load is significantly larger than database load (2 orders of magnitude for Facebook~\cite{bronson2013tao}), even a minor deterioration in hit ratio can yield a prohibitive load on the backend database. Figure~\ref{fig:efficiency:txnCache} shows the experiment results. Each data point is the result of a single run. 

We vary the dependency list size and for each value run the experiment for the two workloads and measure the average values of these metrics. 
{\tCache} is able to reduce inconsistencies significantly. 
For the retailer workload, a single dependency reduces inconsistencies to $56\%$ of their original value, 
two dependencies reduce inconsistencies to $11\%$ of their original value, 
and three to less than $7\%$. 
For the social network workload, with~3 dependencies fewer than $7\%$ of the inconsistencies remain. 

In both workloads there is no visible effect on cache hit ratio, and hence no increased access rate at the database. The reduction in inconsistency ratio is significantly better for the retailer workload. Its topology has a more clustered structure, and so the dependency lists hold more relevant information. 

Next we compared our technique with a simple approach in which we limited the life span (Time To Live, TTL) of cache entries. Here inconsistencies are not detected, but their probability of being witnessed is reduced by having the cache evict entries after a certain period even if the database did not indicate they are invalid. 

We run a set of experiments similar to the \tCache\ ones, varying cache entry TTL to evaluate the efficacy of this method in reducing inconsistencies and the corresponding overhead. 
Compared to \tCache, Limiting TTL has detrimental effects on cache hit ratio, quickly increasing the database workload. By increasing database access rate to more than twice its original load we only observe a reduction of inconsistencies of about $10\%$. 
This is more than twice the rate of inconsistencies achieved by \tCache\ for the retailer workload and only slightly better than the rate of inconsistencies achieved by \tCache\ for the social network workload; and with twice the additional load on the database.

\begin{figure}[!t]
\centering
\includegraphics[width=0.8\linewidth]{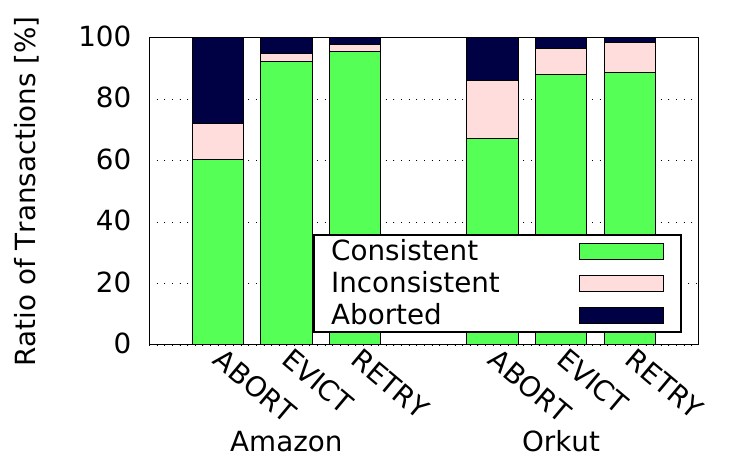} 
\caption{
The efficacy of \tCache\ as a function of the inconsistency handling
strategy for realistic workloads.
}
\label{fig:onInconsistencyReal}
\end{figure} 

            \subsubsection{Detection vs. Prevention} \label{sec:onInconsistencyRealistic} 

Figure~\ref{fig:onInconsistencyReal} compares the efficacy of the ABORT, EVICT and RETRY policies with the Amazon and Orkut workloads.
In these experiments we use dependency lists of length~3.  Just as with the synthetic workload, evicting conflicting transactions is an effective way of invalidating stale objects that might cause problems for future transactions.


The effects are more pronounced for the well-clustered Amazon workload. 
With the Amazon workload, ABORT is able to detect $70\%$ of the inconsistent transactions, whereas with the less-clustered Orkut workload it only detects $43\%$. 
In both cases EVICT reduces uncommittable transactions considerably, relative to their value with ABORT~--- $20\%$ with the Amazon workload and $36\%$ with Orkut. In the Amazon workload,
RETRY further reduces this value to $11\%$ of its value with ABORT.

    \section{Related Work} 

\paragraph{Scalable Consistent Databases} 


Recent years have seen a surge of progress in the development of scalable object stores that support transactions. 
Some systems such as~\cite{LFKA11, SPAL11, LPC12, bailis2014ramp, Xie14Salt} export novel consistency definitions that allow for effective optimizations. 
Several recent systems implement full fledged atomicity while preserving the system's scalability with a wide variety of workloads.
Google's Spanner utilizes accurate clock synchronization. 
Tango~\cite{balakrishnan2013tango} by Balakrishnan et al.\ is constructed on top of the scalable Corfu~\cite{balakrishnan2012corfu} log. 
Eyal et al.~\cite{eyal2013ordering} utilize a large set of independent logs. 
Escriva et al.~\cite{escriva2013warp} use DHT-based locking. 
Zhang et al.~\cite{zhang2013transaction} use lock chains and assume transactions are known in advance. 
These methods all scale well and in many cases allow databases to accept loads similar to those handled by non-transactional databases. 
Nevertheless, they are not expected to disrupt the prevailing two-tier structure; \emph{caches remain invaluable}. 

Note that we are addressing the problem of read-only incoherent caches that respond to queries without access to the backend database. Previous work on coherent caches, e.g.~\cite{franklin1997transactional, carey1994shore,adya1995loosely}, supports transactions using locks or communication with the database on each transaction. These techniques are not applicable in our scenario. 

\paragraph{Consistent Caching} 
Much work has been done on creating consistent caches for web
servers~\cite{CID99,YBS99,ZY01,AO06,OracleAS}, distributed file systems~\cite{Kent86,VEA96}, Key-Value Stores~\cite{Memcached,bronson2013tao,nishtala2013scaling} and higher level objects~\cite{GBHC00,BCF04}.
Such systems consider only one object at a time, and only individual read
and write operations, as they do not support a transactional interface.
There are few if any multi-object or multi-operation
consistency considerations.
These systems generally try to avoid staleness through techniques such as
Time-To-Live (TTL), invalidation broadcasts, and leases.
\emph{Our work considers multi-object transactional consistency of cache access}.

\paragraph{Transactional Caching} 
Early work on scalable database caching mostly ignored transactional
consistency~\cite{LKM02}.
Since then, work has been done on creating consistent caches for databases.
TxCache~\cite{ports2010txCache} extends a centralized database with support
for caches that provide snapshot isolation semantics,
albeit the snapshots seen may be stale.
To improve the commit rate for read-only transactions, they use
\emph{multiversioning}, where the cache holds several versions of an object
and enables the cache to choose a version that allows a transaction to commit. This technique could also be used with our solution.
Perez-Sorrosal et al.~\cite{PPJK07,perez2011elastic} also support snapshot isolation, but
can be used with any backend database,
including ones that are sharded and/or replicated.
JBossCache~\cite{JBossCache} provides a transactionally consistent cache
for the JBoss middleware.
Both JBossCache and \cite{LR03} support transactions on cached
Enterprise JavaBeans.
\cite{BFG06} allows update transactions to read stale data out of caches
and provide bounds on how much staleness is allowed.
These techniques require fast communication between the cache and the
database for good performance.  In contrast, \emph{in our work
caches are asynchronously updated (or invalidated)}, which is how
caches currently work in large multi-regional clouds.

    \section{Future Directions}

The dependency list sizes for all objects in \tCache\ are currently all of the same maximum length. This may not be optimal. For example, if the workload accesses objects in clusters of different sizes, objects of larger clusters call for longer dependency lists. Once appropriate real workloads are available, it may be possible to improve performance by dynamically changing per-object dependency list sizes, balancing between objects to maintain the same overall space overhead.   Another option is to explore an approach in which each type of object would have its own dependency list bound; this could work well if a system has multiple classes of objects, all clustered but with different associated clustering properties.

At present, \tCache\ is semantics-agnostic and treats all objects and object relations as equal, using an LRU policy  to trim the list of dependencies. However, there may be cases in which the application could explicitly inform the cache of relevant object dependencies, and those could then be treated as more important and retained, while other less important ones are managed by some other policy such as LRU. For example, in a web album the set of pictures and their ACL is an important dependency whereas occasional tagging operations that relate pictures to users may be less important. It may be straightforward to extend the cache API to allow the application to specify such dependencies and to modify \tCache\ to respect them. 

    \section{Conclusion} 

Existing large-scale computing frameworks make heavy use of edge caches to reduce client latency, but this form of caching has not been available for transactional applications.  We believe this is one reason that transactions are generally not considered to be a viable option in extremely large systems.

We defined cache-serializability, a variant of serializability that
is suitable for incoherent caches, which cannot communicate with the 
backend database on every read access.  We then presented \tCache, an
architecture for controlling transaction consistency with caches.
The system extends the edge cache by allowing it to offer a
transactional interface.  We believe that \tCache\ is the first
transaction-aware caching architecture in which caches are updated
asynchronously.
In particular, a lookup request only requires a round-trip to the
database in case there is a cache miss~--- there is no additional
traffic and delays to ensure cache coherence.

{\tCache} associates dependency information with cached database
objects, while leaving the interaction between the backend systems
and the cache otherwise unchanged.  This information includes version
identifiers and bounded-length dependency lists.  With even a modest
amount of additional information, we show that inconsistency can
be greatly reduced or even completely eliminated in some cases.
With unbounded resources, we proved that \tCache's algorithm implements cache-serializability. 

\tCache\ is intended for clustered workloads, which are common in social networks, product relationships, mobile applications with spatial locality, and many other
cloud computing applications. 
Our experiments demonstrate \tCache\ to be extremely effective in realistic workloads based on datasets from Amazon and Orkut. Using dependency lists of size~3, {\tCache} detected $43-70\%$ of the inconsistencies, and was also able to increase consistent transaction rate by $33-58\%$ with only nominal overhead on the database.
Our experiments with synthetic workloads showed that \tCache's efficacy depends on the clustering level of the workload. \tCache\ also responds well to change, and in particular, 
we showed that the system rapidly adapts to workloads where data clustering evolves over time.


\bstctlcite{ieeeBSTCTLMaxNames} 
\bibliographystyle{IEEEtran} 
\bibliography{txnCache} 

\appendix

\newcommand{\sigmaUpdate}{\ensuremath{ \sigma_\textit{update} }}

    \section{Consistency} \label{app:theory} 

We now prove Theorem~\ref{thm:consistency}. 

\consistencyThm*

Since we assume that the transactional DB is serializable, the operations in an execution of update transactions \sigmaUpdate\ can be serialized as some serial execution~$\pi$. The next claim trivially follows from the definition of the database dependency list specification: 

\begin{claim} \label{clm:dependencies}
If $\pi$ is a serialization of the update transactions of an execution \sigmaUpdate, then, at every step in $\pi$, the version dependencies of every object match those stored in its dependency list. 
\end{claim} 

To prove Theorem~\ref{thm:consistency}, we first describe a routine for placing a read-only transaction from a cache server in a serialization of a subset of $\sigma$, to form a serialization of both the update transaction and the read-only transaction. 

        \subsection{Permutation routine} 

Let $\sigma$ be an execution of the \tCache\ system, and denote by \sigmaUpdate\ the projection of $\sigma$ on the set of database update transactions. 
Let~$T$ be a read-only cache transaction that reads objects $o_1, o_2, \dots, o_n$ with versions $v_1,v_2, \dots, v_n$, respectively. 
Take any serialization $\pi$ of \sigmaUpdate\ (one exists according to Claim~\ref{clm:dependencies}) and consider the first time when all the objects the transaction reads are at a version at least as large as the versions that~$T$ reads. 
At this time at least one object read by $T$, the last written according to $\pi$, has the correct version, but others might not. Assume without loss of generality that the last version written is~$v_n$ of object $o_n$ at step $t$ of $\pi$. 
Denote by $t'$ the latest time at which a wrong version (not the one read by $T$) is written, and assume WLOG it is version $v_{n-1} + k$ of object $o_{n-1}$ (rather than the desired version $v_{n-1}$) for some $k \ge 1$. 

We now describe a single iteration of the routine. 
Consider the transactions between $t'$ and $t$ (inclusive). Divide these transactions into three sets:

\begin{description} 
\item[Set 1] Transactions dependent on the transaction at~$t'$ (including~$t'$). 
\item[Set 2] Transactions on which $t$ is dependent (including~$t$). 
\item[Set 3] Transactions that do not belong to either group. 
\end{description} 

The following Lemma states that there is no dependency among objects in sets~1 and~2, and hence there is no intersection between the sets. 

\begin{lemma} 
Sets~1 and~2 are independent. 
\end{lemma} 

\begin{proof} 
If they were dependent, then version $v_n$ of object $o_n$ depends on version $v_{n-1} + k$ of object $o_{n-1}$, and this dependency is reflected in their \tCache\ dependency lists, because they are unbounded. However, transaction~$T$ has read version $v_{n-1}$ of object $o_{n-1}$, which is older than $v_{n-1} + k$. The read of the stale version~$v_{n-1}$ of $o_{n-1}$ would have been detected by \tCache\ and the transaction would have been aborted. Therefore the assumption is wrong, and the sets are indeed independent. 
\end{proof} 

Set~3, perhaps an empty set, is unrelated to sets~1 and~2 by definition. We therefore switch sets~1 and~2, and place set~3 right after them, maintaining a serialization of \sigmaUpdate. 

For example, consider the following serialization: 
($X_i$ denotes a transaction $X$ in set $i$):

\begin{center} 
\noindent
\begin{tabularx}{0.8\linewidth}{|*{7}{Y|}} 
\hline
\cellcolor{white} $A_1$ &
\cellcolor{white} $B_3$ &
\cellcolor{white} $C_1$ &
\cellcolor{white} $D_1$ &
\cellcolor{white} $E_3$ &
\cellcolor{white} $F_2$ &
\cellcolor{white} $G_2$ 
\\
\hline
\end{tabularx}
\end{center} 
After the permutation, we obtain:
\begin{center} 
\noindent
\begin{tabularx}{0.8\linewidth}{|*{7}{Y|}} 
\hline
\cellcolor{white} $A_1$ &
\cellcolor{white} $C_1$ &
\cellcolor{white} $D_1$ &
\cellcolor{white} $F_2$ &
\cellcolor{white} $G_2$ & 
\cellcolor{white} $B_3$ &
\cellcolor{white} $E_3$ 
\\
\hline
\end{tabularx}
\end{center} 

Performing this permutation is one iteration of the routine. 
We repeat this iteration forming a series of permutations. Each permutation is a serialization of~\sigmaUpdate, and each permutes a range of the transactions with respect to the result of the previous iteration. 
In each iteration the right end of the permuted range is smaller than the left end of range permuted by the previous iteration, as one or more of the objects is closer to the value read by~$T$. Eventually we therefore reach a permutation where at the chosen time all read objects are at their correct versions. We place $T$ there to obtain the desired serialization of the update transactions and $T$. 

        \subsection{\tCache\ Consistency} 

We proceed to prove Theorem~\ref{thm:consistency}. 

\begin{proof} 
Let $\sigma$ be an execution of the \tCache\ system, and denote by \sigmaUpdate\ the projection of $\sigma$ on the set of database update transactions. Denote by $T_1, T_2, \dots, T_m$ a set of read-only transactions performed through a single \tCache\ server. 

If the read sets of two transactions include the same object $o$, we say the one that read a larger version of $o$ depends on the other. All transactions access the same cache, and the cache is unbounded.  Therefore, values are only replaced by newer versions, so it is easy to see that there are no cycles such that two transactions depend on one another. 
The dependency graph therefore describes a partial order of the read-only transactions, and we choose an arbitrary total ordering that respects this partial order. Assume WLOG the order is $T_1, T_2, \dots, T_m$. 

We take an initial arbitrary serialization~$\pi_0$ of~$\sigma$ and permute it according to the route above to place~$T_1$, the first read-only transaction. The result is a permutation~$\pi_1'$ that includes~$T_1$. Then, we take all transactions that precede~$T_1$ in~$\pi_1'$ although $T_1$ does not depend on them, and place them after~$T_1$. We call this permutation~$\pi_1$. 

Next we place~$T_2$ by permuting~$\pi_1$. If $T_2$ can be placed immediately after $T_1$, we place it there to form~$\pi_2$. 
If $T_2$ is independent of $T_1$ then all its preceding transactions (according to the dependency graph) are unrelated to $T_1$ and are therefore located after it. The permutations required are therefore after $T_1$'s location. Finally, if $T_2$ depends on $T_1$, all relevant update transactions are located after $T_1$ in $\pi_1$, and therefore the permutations required are all after $T_1$'s location. Since in all cases the permutations are after $T_1$'s location in $\pi_1$, they do not affect the correctness of $T_1$'s placement. We take the resulting permutation that we call $\pi_2'$, and move all transactions that neither $T_2$ nor $T_1$ depend on to right after $T_2$. The resulting permutation is $\pi_2$. 

We repeat this process until we place all read-only transactions, forming~$\pi_m$. This is a serialization of the update transactions in $\sigma$ and all read-only transactions that accessed the same cache. We have therefore shown that in any execution of \tCache\ the update transactions can be serialized with read-only transactions that accessed a single cache, which means that \tCache\ implements cache serializability. 
\end{proof}


\end{document}